\documentclass[12pt]{article}
\newcommand{\version}{\today}

\usepackage{amsthm,amsfonts,amsmath, amscd,bbold}

\swapnumbers
                              
\theoremstyle{plain}
\newtheorem{thm}{THEOREM}[section]
\newtheorem{lm}[thm]{LEMMA}

\theoremstyle{definition}

\theoremstyle{remark}
\newcommand{\upchi}{\raise1pt\hbox{$\chi$}}
\newcommand{\R}{{\mathord{\mathbb R}}}

\newcommand{\N}{{\mathord{\mathbb N}}}

\newcommand{\tr}{{\rm Tr}}
\renewcommand{\|}{{\Vert}}
\numberwithin{equation}{section}
\renewcommand{\ln}{\log}

\newcommand{\un}{{\rm 1\kern -2.5pt l}}

\usepackage{color}
\usepackage[colorlinks=true]{hyperref}
\hypersetup{urlcolor=blue, citecolor=red, linkcolor=blue}

\begin{document}
\markboth{\scriptsize{CLL \version}}{\scriptsize{CLL February 26, 2015}}
\def\H{\mathcal{H}}
\def\cT{\mathcal{T}}

\title{{\sc On a Quantum Entropy Power  Inequality of Audenaert, Datta and Ozols}}
\author{
 Eric A. Carlen$^{1}$, Elliott H. Lieb$^{2}$ and Michael Loss$^{3}$ \\
\small{$^{1}$ Department of Mathematics, Hill Center,}\\[-6pt]
\small{Rutgers University, 110 Frelinghuysen Road Piscataway NJ 08854-8019 USA}\\[-6pt]
\small{$^{2}$ Departments of Mathematics and Physics, Jadwin Hall, Princeton University}\\[-6pt]
\small{ P.~O.~Box 708, Princeton, NJ  08542.}\\[-6pt]
\small{$^{3}$ School of Mathematics, Georgia Tech, Atlanta GA 30332}\\[-6pt]
 }

\date{ March  19, 2016}
\maketitle

\footnotetext [1]{Work partially supported by U.S.
National Science Foundation grant  DMS 1501007.}

\footnotetext [2]{Work partially supported by U.S. National
Science Foundation
grant PHY-1265118 }

\footnotetext [3]{ Work partially supported by U.S. National Science Foundation Grant DMS 1301555
and the German Humboldt Foundation 

 \hfill\break
\copyright\, 2016 by the authors. This paper may be reproduced, in its
entirety, for non-commercial purposes.
}

\begin{abstract} We give a short proof of a recent inequality of Audenaert, Datta and Ozols, and determine
cases of equality.
\end{abstract}
%
%


\maketitle

\section{Introduction}

Let $X$ and $Y$ be independent $\R^n$ valued random variables with probability densities $\rho$ and  $\sigma$
respectively. For any $\theta \in (0,\pi/2) $, a simple calculation shows that the sum $\cos\theta X + \sin \theta Y$ has the density 
\begin{equation}\label{sccon1}
\rho \star_\theta \sigma(x) :=  
\int_\R \rho(x\cos\theta  + y \sin\theta )\sigma(-x\sin\theta  + y\cos\theta ){\rm d}y\ .
\end{equation}
If $U_\theta$ is the orthogonal transformation on $\R^{2n}$ given in block form by 
$U_\theta  = \left[\begin{array}{cc} \cos\theta \mathbb{1}& -\sin\theta\mathbb{1} \\ \sin\theta \mathbb{1} & \phantom{-}\cos\theta \mathbb{1}\end{array}\right]$, we can rewrite (\ref{sccon1}) as
\begin{equation}\label{sccon2}
\rho \star_\theta \sigma(x) :=  
\int_\R  \rho\otimes \sigma (U_\theta (x,y)){\rm d}y\ .
\end{equation}

Every probability density $\rho$ on $\R$ that has a finite second moment  has a well-defined entropy $S(\rho)$ given by
${\displaystyle S(\rho) = - \int_\R \rho \ln \rho(x){\rm d}x}$
with $S(\rho) \in [-\infty, S(\gamma_\rho)]$ where $\gamma_\rho$ is the centered Gaussian density with the same variance as $\rho$.  One form of Shannon's Entropy Power Inequality, first proved by Stam \cite{S}, is that for all $\rho$ and $\sigma$  with finite variance, and all $\theta\in (0,\pi/2)$,
\begin{equation}\label{sccon2}
S(\rho \star_\theta \sigma) \geq \cos^2\theta S(\rho) + \sin^2\theta S(\sigma)\ .
\end{equation}
For a different proof, see \cite{L}.

Several authors \cite{ADO,KS} have recently investigated quantum analogs of the entropy power inequality.   Let $\H$ be a separable Hilbert space. 
Let $U$ be any unitary operator on $\H\otimes \H$, and let $\rho$ and let $\rho$ and $\sigma$ be two density matrices on $\H$; i.e., positive trace class operators on $\H$.  Then the operation
$$(\rho,\sigma) \mapsto  \tr_2[ U^*(\rho\otimes \sigma) U]$$
where $\tr_2$ is the partial trace over the second factor on $\H\otimes \H$, provides general quantum analog of the classical  scaled convolution operation in (\ref{sccon2}).  To get a closer analog, for which one can prove an analog of the Entropy Power Inequality (\ref{sccon2}), one must make a choice of $U$ that is somehow analogous to the classical rotation in 
(\ref{sccon2}).  There is a natural way to do this for a system of $n$ bosons, and the 
quantum analog of (\ref{sccon2}) was proved in this setting by K\"onig and Smith \cite{KS}, using a quantum implementation of the method of Stam \cite{S}. In their setting, the Hilbert space $\H$ is infinite dimensional.  More recently, Audenaert, Datta and Ozols  have sought and proved \cite{ADO} an analog of (\ref{sccon2})  that does not require the structure associated to a system of $n$ bosons for its formulation, and is valid for any pair of density matrices on any  separable Hilbert space $\H$.

Let  $\H$ be any separable Hilbert space, possibly but not necessarily  finite dimensional.  Define the {\em swap operator} $S$ on $\H\otimes \H$ by 
$S(u\otimes v) = v\otimes u$.  $S$ is self adjoint and unitary, and for each $\theta\in [0,\pi/2]$,
$\cos\theta \mathbb{1}_{\H\otimes \H} + i \sin\theta S$ is unitary. To simplify the notation in what follows, we define  $ \sqrt t := \cos\theta$ so that for $\theta\in (0,\pi/2)$,
$\sin\theta = \sqrt{1- t}$. 
For each $t\in [0,1]$, define
\begin{equation}\label{partialsw}
U_t =  \sqrt t\,  \mathbb{1}_{\H\otimes \H} + i \sqrt{1- t}\, S\ .
\end{equation}
Then  as $t$ varies between $0$ and $1$, 
$U_t$ provides a continuous unitary interpolation between the identity and the swap operator $S$. Hence, for 
$t\in (0,1)$, $U_\theta$ may be thought of as a {\em partial swap}.  
See \cite{ADO} for the phyical context.  For any two density matrices $\rho$ and $\sigma$ on $\H$, and any $t\in [0,1]$,
define
\begin{equation}\label{partialswcon}
\rho\star_t \sigma  =  \tr_2 ( U_t (\rho\otimes \sigma) U^*_t)\ .
\end{equation}
A straightforward computation yields the explicit formula:
\begin{equation}\label{partialswcon2}
\rho \star_t \sigma = t \rho + (1-t)\sigma + i\sqrt{t (1-t)} [\rho,\sigma]\ .
\end{equation}

For a density matrix $\rho$ on $\H$, let $S(\rho) := -\tr[\rho\ln \rho]$ be the von Neumann entropy of $\rho$. Audenaert, Datta and Ozols  prove \cite{ADO} the following analog of (\ref{sccon2}):
\begin{equation}\label{ado}
S(\rho\star_t \sigma) \geq t S(\rho) + (1-t)S(\sigma)\ . 
\end{equation}

In \cite{ADO}, (\ref{ado}) is by means of  a {\em majorization inequality}; see \cite{MO}:
Given two monotone non-increasing sequence $\{\kappa_j\}$ and $\{\lambda_j\}$ of non-negative numbers with $\sum_{j=1}^\infty \kappa_j = 
\sum_{j=1}^\infty \lambda_j < \infty$, then  $\{\kappa_j\}$ {\em is majorized  by} $\{\lambda_j\}$
 if and only if for all $k\in \N$, 
$\sum_{j=1}^k \kappa_j \leq \sum_{j=1}^k\lambda_j$, in which case we write $\{\kappa_j\} \prec \{\lambda_j\}$.

A theorem of Hardy, Littlewood and P\'olya \cite{HLP} and Karamata \cite{K}, extended to infinite sequences in
\cite{GM}, says that $\{\kappa_j\} \prec \{\lambda_j\}$ if an only if there is a doubly stochastic matrix $S$ 
($S_{i,j} \geq 0$ and  $\sum_{j=1}^\infty S_{i,j} = \sum_{j=1}^\infty S_{i,j}  = 1$ for all $i,j$)  such that
$\kappa_i = \sum_{j=1}^\infty S_{i,j}\lambda_j$ for all $j$. 
Then  by the {\em strict} concavity of $h(x) = -x \ln x$ on $[0,1]$ we have the following:
 {\em When $\{\kappa_j\}$ and  $\{\lambda_j\}$ are two non-negative non-increasing sequences with
$\sum_{j=1}^\infty  \kappa_j = \sum_{j=1}^\infty \lambda_j = 1$, and $\{\kappa_j\} \prec \{\lambda_j\}$, 
then $-\sum_{\j=1}^\infty \kappa_j\ln   \kappa_j \geq  -\sum_{\j=1}^\infty \lambda_j\ln \lambda_j$,
and when $\sum_{\j=1}^\infty \lambda_j\ln \lambda_j < \infty$, 
there is equality if and only if $\{\kappa_j\} = \{\lambda_j\}$.} For any non-negative compact operator $A$ on $\H$, let $\lambda_j(A)$ denote the $j$th largest eigenvalue of $A$,
with the eigenvalues repeated according to their multiplicity.  
The first part of the following theorem is proved in \cite{ADO} by a much longer argument. 

\begin{thm}\label{epiq}  For any two density matrices $\rho$ and $\sigma$ on $\H$, and any $t\in [0,1]$, let 
$\rho \star_t\sigma$ be given by (\ref{partialswcon2}). Then
\begin{equation}\label{epiq2}
\{\lambda_j(\rho\star_t \sigma)\} \prec \{ t\lambda_j(\rho) + (1-t)\lambda_j(\sigma)\}\ . 
\end{equation}
Moreover, $\{ t\lambda_j(\rho) + (1-t)\lambda_j(\sigma)\} = \{\lambda_j(\rho\star_t \sigma)\}$
if and only if  there is an orthonormal basis $\{\phi_j\}$ of $\H$ such that $\rho\phi_j = \lambda_j(\rho)\phi_j$ and $\sigma
 = \lambda_j(\sigma)\phi_j$ for each $j$. 
\end{thm} 

By Theorem~\ref{epiq}  and what we have said about majorization and $h(x) = -x\ln x$,
\begin{multline}\label{sec}
S(\rho\star_t \sigma) \geq -\sum_{j=1}^\infty h( t\lambda_j(\rho) + (1-t)\lambda_j(\sigma)\}) \geq \\
-t\sum_{j=1}^\infty h( \lambda_j(\rho) ) -(1-t)\sum_{j=1}^\infty h( \lambda_j(\sigma)\}) = tS(\rho) + (1-t)S(\sigma)\ ,
\end{multline}
and this proves  (\ref{ado}).   If there is equality in (\ref{ado}) and the left side is finite,  
then by Theorem~\ref{epiq}, there is an orthonormal basis $\{\phi_j\}$ of $\H$ such that 
$\rho\phi_j = \lambda_j(\rho)\phi_j$ and $\sigma
 = \lambda_j(\sigma)\phi_j$ for each $j$. 
Since the second inequality in (\ref{sec}) must also be saturated, it must be the case that 
$\lambda_j(\rho) =\lambda_j(\sigma)$ for each $j$, and hence $\rho = \sigma$.  
{\em Thus our statement about cases of equality in Theorem~\ref{epiq} implies that finite equality holds in (\ref{ado})
if and only if $\rho = \sigma$. }

We now give a very short proof of Theorem~\ref{epiq}. 
The second part, on cases of equality, is new. The heart of the matter is the following lemma. 

\begin{lm}\label{lm1} For all non-negative compact contractions $A$ and $B$ on $\H$, and all $t\in (0,1)$, 
let $A \star_t B := tA + (1-t)B + i\sqrt{t(1-t)}[A,B]$. Then
$\lambda_1(A \star_t B) \leq t\lambda_1(A) + (1-t)\lambda_1(B)$.  There is equality if and only if there is a unit
vector $\phi\in \H$ such that $A\phi = \lambda_1(A)\phi$ and  $B\phi = \lambda_1(B)\phi$.
\end{lm}

\begin{proof}  To prove the inequality, it suffices to show that 
\begin{equation}\label{bnd}
[t\lambda_1(A) + (1-t)\lambda_1(B) ]\mathbb{1} - A \star_t B \geq 0\ .
\end{equation}
Define 
\begin{equation}\label{bnd1}
X= (\lambda_1(A) \mathbb{1}  - A) \ ,\qquad  Y = (\lambda_1(B) \mathbb{1}  - B) \quad{\rm and}\quad  Z = \sqrt{t}X + i\sqrt{1-t}Y\ .
\end{equation}
Note that 
$i \sqrt{t (1-t)} [A,B] = i[\sqrt{t}X, \sqrt{1-t}Y]  = -ZZ^* +  t X^2 + (1-t)Y^2$.
Therefore, the left side of (\ref{bnd}) can be written as
$t [ X -X^2] + (1-t)[ Y-Y^2] +  ZZ^*$.
Since $0 \leq X,Y \leq \mathbb{1}$, $X-X^2 \geq 0$ and $Y-Y^2 \geq 0$, and this proves (\ref{bnd}). 

For the cases of equality, suppose that $\phi$ is in the null space of 
$t [ X -X^2] + (1-t)[ Y-Y^2] +  ZZ^*$.
Then $\langle \phi, [ X -X^2] \phi\rangle =0$, and hence either $X\phi  =0$ or $X\phi = \phi$, and likewise for $Y$ in place of $X$.   Then if  either $X\phi = \phi$ or $Y\phi= \phi$,
then $\|Z^*\phi  \| \geq \sqrt{ t(1-t)} \|\phi\|$, so if $\phi \neq 0$, $X\phi = Y\phi = 0$, and this proves the statement about cases of equality. 
\end{proof}

\begin{proof}[Proof of Theorem~\ref{epiq}]
For any $k\in \N$, let $\bigwedge^k\H$ denote the antisymmetric $k$-fold tensor product of $\H$. For any bounded operator $A$ on $\H$,
define the operator $a^{[k]}$  on $\bigwedge^k\H$ by 
$$A^{[k]}(v_1,\wedge  \cdots \wedge v_k) = (Av_1\wedge v_2\wedge \cdots \wedge v_k) + (v_1\wedge A v_2 \wedge \cdots \wedge v_k) + \cdots 
+ (v_1\wedge \cdots \wedge v_{k-1}\wedge Av_k)\ .$$
If $A$ is self adjoint, non-negative  and compact, so is $A^{[k]}$ and  
${\displaystyle 
\lambda_1(A^{[k]}) = \sum_{j=1}^k \lambda_j(A)}$.
In particular, if $A$ is a density matrix, then for all $k\in \N$, $A^{[k]}$ is a non-negative contraction. 
Also, for any bounded operators $A,B$ on $\H$, $[A^{[k]},B^{[k]}] = [A,B]^{[k]}$.   

It now follows from Lemma~\ref{lm1} that for any density matrices $\rho$ and $\sigma$, the sequence  $\{\lambda_j(\rho \star_t \sigma)\}$ is majorized
by the sequence $\{ t\lambda_j(\rho) + (1-t)\lambda_j(\sigma)\}$, and moreover, if $\sum_{j=1}^k\lambda_j(\rho\star_t \sigma) = \sum_{j=1}^k   ( t\lambda_j(\rho) + (1-t)\lambda_j(\sigma))$ for each $k$, then $\rho$ and $\sigma$ have a common eigenvector basis $\{\phi_j\}$ such that   $\rho\phi_j = \lambda_j(\rho)\phi_j$ and $\sigma
 = \lambda_j(\sigma)\phi_j$ for each $j$. 

\end{proof}


\begin{thebibliography}{30}

\footnotesize{



\bibitem{ADO} K.~Audenaert, N.~Datta and M.~Ozols, \textit{Entropy power inequality for qdits}.
arXiv:1503.04213v2.

\bibitem{HLP} G.H.~Hardy, J.E.~Littlewood, and G.~P—\'olya, \textit{Some simple inequalities satisfied by convex
functions}, Messenger of Math. {\bf 58}, 145-152 (1929).

\bibitem{K} J.~Karamata, \textit{Sur une in\'egalit\'e relative aux fonctions convexes}, Publ. Math. Univ.
Belgrade {\bf 1}, 145-148 (1932).

\bibitem{KS}  R.~K\"onig,  and G.~Smith,
\textit{The entropy power inequality for quantum  systems,}  IEEE Trans. Info. Thy.  {\bf 60},
1536-1548 (2014).

\bibitem{GM} I.C. Gohberg, A.S. Markus, \textit{Some relations between eigenvalues and matrix elements of linear operators}, Mat.
Sb. {\bf 64} (106)  481-496, 1964 (in Russian); Amer. Math. Soc. Transl. Ser. 2, vol. 52, 1966,  201-216 (in English).

\bibitem{L} E.H.~Lieb, \textit{Proof of an entropy conjecture of Wehrl},  Commun. Math.
Phys., {\bf 62}, no. 1, 35Ð41, (1978).

\bibitem{MO} A.W.~ Marshall. and  I.~Olkin, \textit{ Inequalities: Theory of Majorization and Its Applications}, Academic Press, New York. (1979).

\bibitem{S} A.J.~Stam. \textit{Some inequalities satisfied by the quantities of information of Fisher and Shannon}
Inf. Control {\bf 2}, 101-112 (1959).


}
\end{thebibliography}
 \end{document}